\documentclass[11pt]{siamltex}

\usepackage{graphicx}
\usepackage{amsmath}
\usepackage{amssymb}
\usepackage{url}
\usepackage{tikz}
\usetikzlibrary{arrows,automata,positioning,shadows}
\numberwithin{equation}{section}

\newcommand{\bm}[1]{\mathbf{#1}}
\newcommand{\ttt}[1]{\texttt{#1}}

\newcommand{\e}{{\mathrm{e}}}
\newcommand{\ii}{{\mathrm{i}}}

\newcommand{\R}{{\mathbb{R}}}

\newcommand{\pib}{{\boldsymbol{\pi}}}

\newcommand{\pqser}{\texttt{PQser} }

\DeclareMathOperator{\spectrsort}{\textsf{spectrsort}}
\DeclareMathOperator{\pqtreeNperm}{\textsf{pqtreeNperm}}

\DeclareMathOperator{\graphvisit}{\textsf{graphvisit}}
\DeclareMathOperator{\getconcomp}{\textsf{getconcomp}}
\DeclareMathOperator{\pnode}{\textsf{pnode}}
\DeclareMathOperator{\qnode}{\textsf{qnode}}
\DeclareMathOperator{\lnode}{\textsf{lnode}}
\DeclareMathOperator{\mnode}{\textsf{mnode}}
\DeclareMathOperator{\length}{length}
\DeclareMathOperator{\factorial}{factorial}

\usepackage{ifthen}
\usepackage{float}
\floatstyle{ruled}
\newfloat{algorithm}{htb}{alg}%[section]
\floatname{algorithm}{Algorithm}
\newcounter{algo@row}
\newcounter{algo@rowindent}
\newcommand{\algofont}[1]{\textbf{#1}}% S1
\newcommand{\algonumbersize}[1]{\scriptsize{#1}}% S2
\newcommand{\algopreitem}[1][\arabic{algo@row}]{\texttt{\algonumbersize{#1}}}
\newcommand{\algoitemskip}{\hspace{\value{algo@rowindent}cc}}
\newenvironment{algo}{\vskip.3em\small%
  \begin{list}{\algopreitem\texttt{\algonumbersize{:}}}{%
      \usecounter{algo@row}%
      \setcounter{algo@rowindent}{0}%
      \setlength{\itemindent}{2em}%
      \setlength{\labelwidth}{2em}% indent=1em
      \setlength{\parsep}{0cm}%
    }%
}{
  \end{list}\vskip-.5em
}
\newcommand{\algonewnestedopen}[2]{
  \newcommand{#1}[1][]{%
    \ifthenelse{\equal{##1}{}}{\item}{\item[{\algopreitem[##1]}]}
    \algoitemskip\algofont{#2}%
    \addtocounter{algo@rowindent}{1}%
    \ignorespaces
  }
}
\newcommand{\algonewnestedaux}[2]{
  \newcommand{#1}[1][]{
    \addtocounter{algo@rowindent}{-1}
    \ifthenelse{\equal{##1}{}}{\item}{\item[{\algopreitem[##1]}]}
    \algoitemskip\algofont{#2}%
    \addtocounter{algo@rowindent}{+1}%
    \ignorespaces
  }
}
\newcommand{\algonewnestedclose}[2]{
  \newcommand{#1}[1][]{
    \addtocounter{algo@rowindent}{-1}
    \ifthenelse{\equal{##1}{}}{\item}{\item[{\algopreitem[##1]}]}
    \algoitemskip\algofont{#2}%
    \ignorespaces
  }
}
\newcommand{\algonewcommand}[2]{
  \newcommand{#1}[1][default]{
    \ifthenelse{\equal{##1}{default}}{\item}{\item[{\algopreitem[##1]}]}%
    \algoitemskip\algofont{#2}%
    \ignorespaces
  }%
}
\newcommand{\algonewkeyword}[2]{\newcommand{#1}{\algofont{#2}}}
\algonewcommand{\STATE}{\ignorespaces}
\algonewcommand{\INPUT}{Input: }
\algonewcommand{\pINPUT}{\phantom{Input: }}
\algonewcommand{\COMPUTE}{Compute: }
\algonewcommand{\OUTPUT}{Output: }
\algonewcommand{\pOUTPUT}{\phantom{Output: }}
\algonewnestedopen{\IF}{if }
\algonewnestedaux{\ELSEIF}{else if }
\algonewnestedaux{\ELSE}{else }
\algonewnestedclose{\ENDIF}{end if }
\algonewnestedopen{\FOR}{for }
\algonewnestedclose{\ENDFOR}{end for }
\algonewnestedopen{\WHILE}{while }
\algonewnestedclose{\ENDWHILE}{end while }
\algonewcommand{\BREAK}{break}%
\algonewkeyword{\For}{for }%
\algonewkeyword{\To}{to }%
\algonewkeyword{\Do}{do }%
\algonewkeyword{\If}{if }%
\algonewkeyword{\Then}{then }%
\algonewkeyword{\Else}{else }%
\algonewkeyword{\End}{end }%
\algonewkeyword{\AND}{and }%
\algonewkeyword{\True}{true }%
\algonewkeyword{\False}{false }%
\algonewkeyword{\Call}{call }%
\algonewkeyword{\Function}{function }%
\algonewkeyword{\irbleigs}{irbleigs }%
\algonewkeyword{\tridiag}{tridiag}%
\algonewkeyword{\reorth}{reorth}%
\algonewkeyword{\distinct}{distinct }%
\algonewkeyword{\fiedvecs}{fiedvecs }%
\begin{document}

%\pagestyle{myheadings}
%\markboth{A. Concas, C. Fenu, and G. Rodriguez}{A spectral algorithm for the seriation problem in archaeology}

\title{\pqser: a Matlab package for spectral seriation}

\author{A.~Concas\thanks{Dipartimento di Matematica e Informatica,
Universit\`a di Cagliari, viale Merello 92, 09123 Cagliari, Italy. E-mail:
\texttt{anna.concas@unica.it}, \texttt{rodriguez@unica.it}.
Research supported by INdAM-GNCS.}
\and C.~Fenu\thanks{AICES Graduate School, RWTH Aachen University Schinkelstrasse 2a, 52062 Aachen, Germany. E-mail: \texttt{fenu@aices.rwth-aachen.de},}
\and G. Rodriguez\footnotemark[1]}

\pagestyle{myheadings}
\markboth{A.~Concas, C.~Fenu and G.~Rodriguez}{A Matlab package for spectral seriation}

%\date{}

\maketitle

%\title{\pqser: a Matlab package for spectral seriation\thanks{Research
%supported in part by INdAM-GNCS.}}
%%\subtitle{Do you have a subtitle?\\ If so, write it here}
%%\titlerunning{Short form of title}        % if too long for running head
%
%\footnotetext{Version \today}
%
%\author{Anna Concas 
%\and Caterina Fenu
%\and Giuseppe Rodriguez}
%
%\institute{A. Concas and G. Rodriguez \at
%Dipartimento di Matematica e Informatica, Universit\`a di Cagliari \\
%viale Merello 92, 09123 Cagliari, Italy. 
%\email{anna.concas@unica.it|rodriguez@unica.it}
%\and 
%C. Fenu \at
%AICES Graduate School, RWTH Aachen University
%Schinkelstrasse 2a, 52062 Aachen, Germany. 
%\email{fenu@aices.rwth-aachen.de}
%% The correct dates will be entered by the editor
%
%\maketitle

\begin{abstract} 
Seriation is an important ordering problem which consists of finding
the best ordering of a set of units whose interrelationship is defined by a
bipartite graph.
It has important applications in, e.g., archaeology, anthropology, psychology,
and biology.
This paper presents a Matlab implementation of an algorithm for spectral
seriation by Atkins et al., based on the use of the Fiedler vector of the
Laplacian matrix associated to the problem, which encodes the set of
admissible solutions into a PQ-tree.
We introduce some numerical technicalities in the original algorithm to
improve its performance, and point
out that the presence of a multiple Fiedler value may have a substantial
influence on the computation of an approximated solution, in the presence of
inconsistent data sets.
Practical examples and numerical experiments show how to use the toolbox to
process data sets deriving from real-world applications.

\end{abstract}



%%%%%%%%%%%%%%%%%%%%%%%%%%%%%%%%%%%%%%%%%%%%%%%%%%%%%%%%%%%%%%%%%%%%%
\section{Introduction}\label{sec:intro}

\emph{Seriation} is an important ordering problem
whose aim is to find the best enumeration order of a set of units, according
to a given correlation function. The desired order can be characteristic of the
data, a chronological order, a gradient or any sequential structure of the
data. 
%We will state the seriation problem from the mathematical point of view,
%by considering it as the arrangement of units in a sequence, according to a
%certain gradient.

The concept of seriation has been formulated in many different ways and appears
in various fields, such as archaeology, anthropology, psychology, and 
biology~\cite{brusco06,genome,matharcheo,mirkin1984}.
In this paper we use the archaeological setting as a metaphor for the seriation
problem.
An important aim of archaeological investigation is to date excavation
sites on the basis of found objects and determine their relative chronology,
i.e., a dating which indicates if a given site is chronologically preceding or
subsequent to another.
In general, relative chronologies are devoid of a direction, in the sense that
the units are placed in a sequence which can be read in both directions.
Relative dating methods can be used where absolute dating methods, such as
carbon dating, cannot be applied.

The available data are usually represented by a \emph{data matrix}, in which
the rows are the archaeological units (e.g., the sites) and the columns
represent the types (the archaeological finds).
Each unit is characterized by the presence of certain artefacts, which are in
turn classified in types. In~\cite{ps05}, the authors refer to the data matrix as either \emph{incidence matrix} or \emph{abundance
matrix}, depending on the archaeological data representation. In the first case, the data are reported by using a binary
representation, i.e., an element in the position $(i,j)$ is equal to $1$ if
type $j$ is present in the unit $i$, and $0$ otherwise.
In the second second case, the data matrix reports the number of objects
belonging to a certain type in a given unit, or its percentage. In this paper,
we will follow the usual terminology used in \emph{complex networks theory} and
we will refer to a binary representation as an \emph{adjacency matrix}, an
example of which is given in Table~\ref{tab:adjacency}. More details can be
found in Section~\ref{sec:mathback}. If
the data matrix represents types of found objects as columns and the locations
(graves, pits, etc.) in which they are found as rows, we can find a
chronological order for the locations by assuming that the types were produced,
or were ``fashionable'', only for a limited period of time. In the light of
this assumption, the purpose of determining a relative chronology results in
obtaining an ordering of the rows and columns of the data matrix that places
the nonzero entries close to the diagonal of the data matrix.

\begin{table}
\caption{Adjacency matrix for archaeological data originated from female
burials at the Bornholm site, Germany; see~\cite{ps05} and the references
therein. The rows report the names of the tombs, the columns the identification
codes of the found \emph{fibulae}.}
\label{tab:adjacency}
\scriptsize
\begin{center}
\begin{tabular}{rcccccccccccc}
\hline\noalign{\smallskip}
&G3&F27&S1&F26&N2&F24&P6&F25&P5&P4&N1&F23\\
\noalign{\smallskip}\hline\noalign{\smallskip}
\textsf{Mollebakken 2} & 1 & 1 & 1 & 1 & 0 & 0 & 0 & 0 & 0 & 0 & 0 & 0\\
\textsf{Kobbea 11}&0&1&1&0&1&1&0&0&0&0&0&0\\ 
\textsf{Mollebakken 1}&1&1&0&1&1&0&1&1&0&0&0&0\\ 
\textsf{Levka 2}&0&1&1&0&1&0&0&1&1&0&0&0\\ 
\textsf{Grodbygard 324}&0&0&0&0&1&1&0&0&0&1&0&0\\ 
\textsf{Melsted 8}&0&0&1&1&0&0&1&1&0&1&0&0\\ 
\textsf{Bokul 7}&0&0&0&0&0&0&1&1&0&0&1&0\\ 
\textsf{Heslergaard 11}&0&0&0&0&0&0&0&1&0&1&0&0\\ 
\textsf{Bokul 12}&0&0&0&0&0&0&0&1&1&0&0&1\\ 
\textsf{Slamrebjerg 142}&0&0&0&0&0&0&0&0&0&1&0&1\\ 
\textsf{Nexo 6} &0&0&0&0&0&0&0&0&0&1&1&1\\ 
\noalign{\smallskip}\hline
\end{tabular}
\end{center}
\end{table}

A closely related problem is the \emph{consecutive ones problem}
(C1P)~\cite{fulkerson65,or09}, whose aim is to find all the permutations of the
rows of a binary matrix that place the $1$'s consecutively in each column. If
such permutations exist, then the matrix is said to have the \emph{consecutive
ones property} for columns. The equivalent property for rows can be similarly
defined. The problem of verifying if a matrix possesses this property has
applications in different fields, such as computational biology and recognition
of interval graphs~\cite{booth1976testing,cor98}.
The connection between C1P and seriation has been investigated by Kendall 
in~\cite{kendall69}.

The first systematic formalization of the seriation problem was made by Petrie
in 1899~\cite{petrie}, even if the term seriation was used before in
archaeology. The subject was later resumed by Breinerd and
Robinson~\cite{brainerd1951place,robinson1951method}, who also proposed a
practical method for its solution, and by
Kendall~\cite{kendall63,kendall69,kendall70}.
Nice reviews on seriation are given in~\cite{liiv2010},~\cite{ps05}, 
and~\cite{bl02}, where its application in stratigraphy is discussed, 
while~\cite{bb15,matharcheo} describe other applications of mathematics in
archaeology.

Given the variety of applications, some software packages have been developed in
the past to manipulate seriation data. Some of these packages have not undergo a
regular maintenance, and does not seem to be easily usable on modern computers,
like the Bonn Archaeological Software Package (BASP)
(\url{http://www.uni-koeln.de/~al001/}).
A software specifically designed for the seriation problem in bioinformatics
has been developed by Caraux and Pinloche~\cite{permutsoft}, and is available
for free download.

Other implementations of the spectral algorithm from~\cite{atkins1998spectral}
have been discussed in~\cite{fogel2014}, \cite{hahsler2008}, which describes an
R package available at \url{http://cran.r-project.org/web/packages/seriation/},
and~\cite{seminaroti2016}.
The paper~\cite{fogel2015} proposes an interesting method, based on quadratic
programming, aimed at treating ``noisy cases''.

In this paper we present a Matlab implementation of a spectral method for the
solution of the seriation problem which appeared in~\cite{atkins1998spectral},
based on the use of the Fiedler vector of the Laplacian associated to the
problem, and which describes the results in terms of a particular data
structure called a PQ-tree.
We further develop some numerical aspects of the algorithm, concerning the
detection of equal components in the Fiedler vector and the computation of the
eigensystem of the Laplacian associated to a large scale problem.
We also provide a parallel version of the method.
The package, named the \pqser toolbox, also defines a data structure to store
a PQ-tree and provides the Matlab functions to manipulate and visualize it.
Finally, we discuss the implications of the presence of a multiple Fiedler
value, an issue which has been disregarded up to now, and we illustrate some
numerical experiments.

The plan of the paper is the following. Section \ref{sec:mathback} reviews the
necessary mathematical background and sets up the terminology to be used in the
rest of the paper. Section~\ref{sec:pqtrees} describes the data structures
used to store the solutions of the seriation problem.
The spectral algorithm is discussed in Section~\ref{sec:seriation} and the
special case of a multiple Fiedler values is analyzed in
Section~\ref{sec:mult}.
Section~\ref{sec:numexp} reports some numerical results and
Section~\ref{sec:last} contains concluding remarks.

%%%%%%%%%%%%%%%%%%%%%%%%%%%%%%%%%%%%%%%%%%%%%%%%%%%%%%%%%%%%%%%%%%%%%
\section{Mathematical background}\label{sec:mathback}

With the aim of making this paper self-contained, we review some mathematical
concepts that will be used in the following.
We denote matrices by upper case roman letters and their elements by
lower case double indexed letters.

Let $G$ be a simple graph formed by $n$ nodes.
Each entry $f_{ij}$ of the adjacency matrix $F\in\R^{n\times n}$ associated to
$G$ is taken to be the weight of the edge connecting node $i$ to node $j$. 
If the two nodes are not connected, then $f_{ij}=0$.
A graph is unweighted if the weights are either 0 or 1.
The adjacency matrix is symmetric if and only if the graph is undirected.

The (unnormalized) \emph{graph Laplacian} of a symmetric, matrix 
$F\in\R^{n\times n}$ is the symmetric, positive semidefinite matrix 
$$
L=D-F,
$$
where $D=\diag(d_1,\ldots,d_n)$ is the \emph{degree matrix}, whose $i$th
diagonal element equals the sum of the weights of the edges starting from node
$i$ in the undirected network defined by $F$, that is, 
$d_i=\sum_{j=1}^n f_{ij}$.
In the case of an unweighted graph, $d_i$ is the number of nodes connected to
node $i$.

Setting $\bm{e}=[1,\dots,1]^T\in\R^n$, it is immediate to observe that
$$
L\bm{e} = (D-F)\bm{e}= \bm{0},
$$
where $\bm{0}\in\R^n$ is the zero vector.
Hence, $0$ is an eigenvalue of the graph Laplacian with eigenvector $\bm{e}$.

The Gershgorin circle theorem implies all the eigenvalues are non-negative, so
we order them as $\lambda_1=0\leq\lambda_2\leq\dots\leq\lambda_n$, with
corresponding eigenvectors $\bm{v}_1=\bm{e},\bm{v}_2,\dots, \bm{v}_n$.
The smallest eigenvalue of $L$ with associated eigenvector orthogonal to
$\bm{e}$ is called the \emph{Fiedler value}, or the \emph{algebraic
connectivity}, of the graph described by $F$.
The corresponding eigenvector is the \emph{Fiedler
vector}~\cite{fiedler1973algebraic,fiedler1975property,fiedler1989laplacian}.

Alternatively, the Fiedler value may be defined by
$$
\min_{\bm{x}^{T}\bm{e}=0,\ \bm{x}^{T}\bm{x}=1}\bm{x}^{T}L\bm{x}.
$$
Then, a Fiedler vector is any vector $\bm{x}$ that achieves the minimum.

From the Kirchhoff matrix-tree theorem it follows that the Fiedler value is
zero if and only if the graph is not connected~\cite{de2007old}; in particular,
the number of times 0 appears as an eigenvalue of the Laplacian is the number
of connected components of the graph.
So, if the considered adjacency matrix is irreducible, that is, if the
graph is connected, the Fiedler vector corresponds to the first non-zero
eigenvalue of the Laplacian matrix.

%In the application to the seriation problem we consider bipartite graphs.
A \emph{bipartite graph} $G$ is a graph whose vertices can be divided into two
disjoint sets $U$ and $V$ containing $n$ and $m$ nodes, respectively, such
that every edge connects a node in $U$ to one in $V$.
%From~\eqref{adjacency} it follows immediately that the adjacency matrix of a
%bipartite graph is of the form
%\begin{equation}
%A = 
%\begin{bmatrix}
%0_1 & B \\ C & 0_2
%\end{bmatrix},
%\qquad 0_1 \in \R^{n \times n},\; 0_2 \in \R^{m \times m}.
%\end{equation}
%If the graph is undirected  than $C = B^T$. 

In our archaeological metaphor, the nodes sets $U$ (units) and $V$
(types) represent the excavation sites and the found artifacts, respectively.
Then, as already outlined in the Introduction, 
the associated adjacency matrix $A$, of size $n\times m$, is obtained by
setting $a_{ij}=1$ if the unit $i$ contains objects of type $j$, and 0
otherwise. 
If the element $a_{ij}$ takes value different from 1, we consider it as a
weight indicating the number of objects of type $j$ contained in unit $i$, or
their percentage. In this case, we denote $A$ as the \emph{abundance matrix}.

The first mathematical definition of seriation was based on the construction of
a symmetric matrix $S$ known as \emph{similarity
matrix}~\cite{brainerd1951place,robinson1951method}, where $s_{ij}$ describes,
in some way, the likeness of the nodes $i,j\in U$. One possible definition is
through the product $S=AA^T$, being $A$ the adjacency matrix of the problem.
In this case, $s_{ij}$ equals the number of types shared between unit $i$ and
unit $j$. 
The largest value on each row is the diagonal element, which reports the
number of types associated to each unit. By permuting the rows and columns of
$S$ in order to cluster the largest values close to the main diagonal, one 
obtains a permutation of the corresponding rows of $A$ that places closer the
units similar in types.
It is worth noting that this operation of permuting rows and columns of $S$ is
not uniquely defined.

The \emph{Robinson method}~\cite{robinson1951method} is a statistical technique
based on a different similarity matrix. It is based on the concept that each
type of artifact used in a certain period eventually decreases in popularity
until it becomes forgotten. This method is probably the first documented
example of a practical procedure based on the use of the similarity matrix, so
its description is interesting in a historical perspective.

The method, starting from an abundance matrix $A \in \R^{n \times m}$ whose
entries are in percentage form (the sum of each row is 100), computes the
similarity matrix $S$ by a particular rule, leading to a symmetric matrix of
order $n$ with entries between $0$ (rows with no types in common) and $200$,
which corresponds to units containing exactly the same types.
Then, the method searches for a permutation matrix $P$ such that $PSP^T$ has
its largest entries as close as possible to the main diagonal. The same
permutation is applied to the rows of the data matrix $A$ to obtain a
chronological order for the archaeological units. Since, as already remarked,
the sequence can be read in both directions, external information must be used
to choose an orientation.

The procedure of finding a permutation matrix $P$ is not uniquely specified.
One way to deal with it is given by the so called \emph{Robinson's form}, which
places larger values close to the main diagonal, and lets off-diagonal entries
be nonincreasing moving away from the main diagonal.
More in detail, a symmetric matrix $S$ is in Robinson's form, or is an
R-matrix, if and only if
\begin{eqnarray}
s_{ij}\leqslant s_{ik}, \quad \text{if } j\leqslant k\leqslant i, 
\label{rmatrix1} \\
s_{ij}\geqslant s_{ik}, \quad \text{if } i\leqslant j\leqslant k.
\label{rmatrix2} 
\end{eqnarray}
A symmetric matrix is pre-$R$ if and only if there exists a simultaneous
permutation of its rows and columns which transforms it in Robinson's form, so
it corresponds to a well-posed ordering problem.
For other interesting references on R-matrices and their detection, 
see~\cite{chepoi1997,laurent2017lex,laurent2017,prea2014,seston2008}.

%The following matrices, excerpted from~\cite{ps05}, are examples of $R$,
%not-$R$ and pre-$R$ matrices:
%\begin{equation*}
%\begin{bmatrix}
%6 & 4 & 2 & 2 \\
% 4 & 8 & 5 & 3 \\
% 2 & 5 & 9 & 4 \\
%2 & 3 & 4 & 7 
%\end{bmatrix}\ R\text{-form}, \qquad
%\begin{bmatrix}
%6 & 4 & 9 & 2 \\
% 4 & 8 & 5 & 3 \\
% 9 & 5 & 9 & 4 \\
%2 & 3 & 4 & 7 
%\end{bmatrix}\ \text{not-}R, \qquad
%\begin{bmatrix}
%9 & 2 & 5 & 4\\
%2 & 6 & 4 & 2 \\
%5 & 4 & 8 & 3\\
%4 & 2 & 3 & 7
%\end{bmatrix}\ \text{pre-}R.
%\end{equation*}
%Note that the last matrix, which is pre-R, is obtained by applying the permutation $P = [\bm{e}_2, \bm{e}_3, \bm{e}_1, \bm{e}_4]$ simultaneously to the rows and columns of the first matrix. Understanding if a
%given symmetric matrix is pre-R is often difficult, in fact not all symmetric
%matrices can be brought to the Robinson's form.

\section{PQ-trees}\label{sec:pqtrees}

A \emph{PQ-tree} is a data structure introduced by Booth and
Lueker~\cite{booth1976testing} to encode a family of permutations of a
set of elements, and solve problems connected to finding admissible permutations
according to specific rules.

A PQ-tree $T$ over a set $U = \{u_1,u_2,\dots,u_n\}$ is a rooted tree whose
leaves are elements of $U$ and whose internal (non-leaf) nodes are
distinguished as either P-nodes or Q-nodes. 
The only difference between them is the way in which their children are
treated: for Q-nodes only one order and its reverse are allowed, whereas in the
case of P-node all possible permutations of the children leaves are permitted.
The root of the tree can be either a P or a Q-node.
%The children of a P-node may occur in any order, and so correspond to the set
%containing all their permutations, while those of a Q-node may be ordered only 
%left-to-right or right-to-left. 

We will represent graphically a P-node by a circle, and a Q-node by a
rectangle. 
The leaves of $T$ will be displayed as triangles, and labeled by the elements
of $U$. The frontier of $T$ is one possible permutation of the elements of $U$,
obtained by reading the labels of the leaves from left to right.

We recall two definitions from~\cite{booth1976testing}.

\begin{definition}\label{def:proper}
A PQ-tree is \textit{proper} when the following conditions hold:
\begin{itemize}
\item[i)] every $u_{i}\in U$ appears precisely once as a leaf;
%This because PQ-trees are supposed to represent permutations of a set, so it
%does not make sense for an element to appear more than once or to not appear at
%all;
\item[ii)] every P-node has at least two children;
%This rules out long chains of nodes having only a single child;
\item[iii)] every Q-node has at least three children. 
\end{itemize}
\end{definition}

As we observed above, the only difference between a P-node and a Q-node is the
treatment of their children, and in the case of exactly two children there is
no real distinction between a P-node and a Q-node.
This justifies the second and third conditions of Definition~\ref{def:proper}.

\begin{definition}\label{def:equiv}
Two PQ-trees are said to be \textit{equivalent} if one can be transformed into
the other by applying a sequence of the following two transformations:
\begin{itemize}
\item[i)] arbitrarily permute the children of a P-node;
\item[ii)] reverse the children of a Q-node.
\end{itemize} 
\end{definition}

A PQ-tree represents permutations of the elements of a set through
admissible reorderings of its leaves.
Each transformation in Definition~\ref{def:equiv} specifies an admissible
reordering of the nodes within a PQ-tree.
For example, a tree with a single P-node represents the equivalence
class of all permutations of the elements of $U$, while a tree with a single
Q-node represents both the left-to-right and right-to-left orderings of the
leaves.
A tree with a mixed P-node and Q-node structure represents the equivalence
class of a constrained permutation, where the exact structure of the tree
determines the constraints.
Figure~\ref{fig:pqtree} displays a PQ-tree and the admissible permutations
it represents.

\begin{figure}
\begin{minipage}{.52\textwidth}
\includegraphics[width=\textwidth]{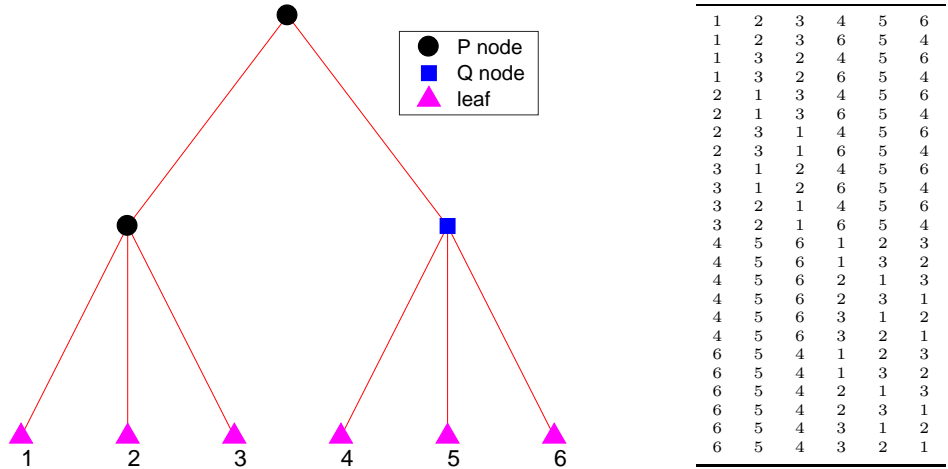}
\end{minipage}
\begin{minipage}{.45\textwidth}
\tiny
\begin{center}
\begin{tabular}{cccccc}
\hline\noalign{\smallskip}
1 & 2 & 3 & 4 & 5 & 6 \\
1 & 2 & 3 & 6 & 5 & 4 \\
1 & 3 & 2 & 4 & 5 & 6 \\
1 & 3 & 2 & 6 & 5 & 4 \\
2 & 1 & 3 & 4 & 5 & 6 \\
2 & 1 & 3 & 6 & 5 & 4 \\
2 & 3 & 1 & 4 & 5 & 6 \\
2 & 3 & 1 & 6 & 5 & 4 \\
3 & 1 & 2 & 4 & 5 & 6 \\
3 & 1 & 2 & 6 & 5 & 4 \\
3 & 2 & 1 & 4 & 5 & 6 \\
3 & 2 & 1 & 6 & 5 & 4 \\
4 & 5 & 6 & 1 & 2 & 3 \\
4 & 5 & 6 & 1 & 3 & 2 \\
4 & 5 & 6 & 2 & 1 & 3 \\
4 & 5 & 6 & 2 & 3 & 1 \\
4 & 5 & 6 & 3 & 1 & 2 \\
4 & 5 & 6 & 3 & 2 & 1 \\
6 & 5 & 4 & 1 & 2 & 3 \\
6 & 5 & 4 & 1 & 3 & 2 \\
6 & 5 & 4 & 2 & 1 & 3 \\
6 & 5 & 4 & 2 & 3 & 1 \\
6 & 5 & 4 & 3 & 1 & 2 \\
6 & 5 & 4 & 3 & 2 & 1 \\
\noalign{\smallskip}\hline
\end{tabular}
\end{center}
\end{minipage}
\caption{On the left, a PQ-tree over the set $U=\{1,\ldots,6\}$; on the right,
the 24 admissible permutations encoded in the tree.}
\label{fig:pqtree}
\end{figure}

The PQ-tree data structure has been exploited in a variety of applications,
from archaeology and chronology reconstruction~\cite{atkins1998spectral} to
molecular biology with DNA mapping and sequence
assembly~\cite{greenberg1995physical}.
The first problem to which it was applied is the consecutive ones property
(C1P) for matrices~\cite{booth1976testing}, mentioned in
Section~\ref{sec:intro}.
%We recall that, in this case, one seeks to find a permutation of the rows of a
%matrix that places the nonvanishing entries within each column consecutively.

Given a pre-R matrix, the spectral algorithm from~\cite{atkins1998spectral},
that will be discussed in Section~\ref{sec:seriation}, constructs a PQ-tree
describing the set of all the permutations of rows and columns that lead to an
R-matrix.

\subsection{Implementation of PQ-trees}

The \pqser toolbox for Matlab is available as a compressed archive that can be downloaded from the authors webpages (see, e.g., 
\url{http://bugs.unica.it/~gppe/soft/}).
By uncompressing it, the directory \pqser will be created. 
It must be added to Matlab search path, either by the command \ttt{addpath} or
using the graphical interface menus.
The sub-directory \texttt{demo} contains a tutorial for the toolbox and the
scripts which were used to construct the examples reported in the paper.
The installation procedure and the toolbox content are described in detail in
the README.txt file, which can be found in the main directory.

In the \pqser toolbox, a PQ-tree $T$ is a \ttt{struct} variable (i.e., a
record) composed by two fields.
The first field, \ttt{T.type}, specifies the type of the node, i.e., P, Q,
or a leaf, in the case of a trivial tree.
The second field, \ttt{T.value}, is a vector which provides a list of PQ-trees,
recursively defined.
In the case of a leaf, this field contains the index of the unit it represents.

For example, the graph in Figure~\ref{fig:pqtree} was obtained by the following
piece of code
\begin{quote}
\footnotesize
\begin{verbatim}
v(1) = pnode([1 2 3]);	% create a P-node with three leaves
v(2) = qnode([4 5 6]);	% create a Q-node with three leaves
T = pnode(v);		% create a P node pointing to the previous two nodes
pqtreeplot(T)		% visualize the PQ-tree
\end{verbatim}
\end{quote}
the resulting data structure for the PQ-tree is
\begin{quote}
\footnotesize
\begin{verbatim}
T = 
  struct with fields:
     type: 'P'
    value: [1x2 struct]
\end{verbatim}
\end{quote}
and the permutations encoded in $T$ are computed by
\begin{quote}
\footnotesize
\begin{verbatim}
perms_matrix = pqtreeperms(T)
\end{verbatim}
\end{quote}
These instructions are contained in the script \texttt{graf1.m}, in the
\texttt{demo} sub-directory.

\begin{table}[htb]
\footnotesize
\centering
\begin{tabular}{ll}
\hline
\ttt{pnode}         & create a P-node \\
\ttt{qnode}         & create a Q-node \\
\ttt{lnode}         & create a leaf \\
\ttt{mnode}         & create an M-node \\
\ttt{pqtreeplot}    & plot a PQ-tree \\
\ttt{pqtreeNperm}   & number of admissible permutations in a PQ-tree \\
\ttt{pqtreeperms}   & extract all admissible permutations from a PQ-tree \\
\ttt{pqtree1perm}   & extract one admissible permutation from a PQ-tree \\
\ttt{pqtreegetnode} & extract a subtree from a PQ-tree \\
\ttt{pqtreenodes}   & converts a PQ-tree to Matlab \ttt{treeplot} format \\
\hline
\end{tabular}
\caption{Functions in the \pqser toolbox devoted to the manipulation of
PQ-trees.}
\label{tab:pqfuncs}
\end{table}

The functions intended for creating and manipulating a PQ-tree are listed in
Table~\ref{tab:pqfuncs}.
%As it is customary for graph-like data processing, most of these functions are
%recursive.
The function \ttt{mnode} creates an additional type of node, an M-node, which
is intended to deal with multiple Fiedler values; we will comment on it in
Section~\ref{sec:mult};
\ttt{pqtreegetnode} and \ttt{pqtreenodes} are utility functions for
\ttt{pqtreeplot}, they are not intended to be called directly by the user.
All the functions are documented via the usual Matlab \ttt{help} command, e.g.,
\begin{quote}
\footnotesize
\begin{verbatim}
help pnode
help pqtreeplot
\end{verbatim}
\end{quote}

As an example, we report in Algorithm~\ref{alg:pqtreeNperm} the structure of
\ttt{pqtreeNperm}, a function which returns the number $N$ of all the
permutations contained in the tree whose root $T$ is given in input.
In the particular case of a leaf, only one permutation is possible 
(line~\ref{l1}--\ref{l2}).
Otherwise, we consider the vector $\bm{c}$ of size $k$, containing the children
nodes of the root of T (line~\ref{l5}).
The algorithm calls itself recursively on each component of $\bm{c}$
(line~\ref{l8}).
In the case of a Q-node the number of permutations is doubled, because only one
ordering and its reverse are admissible, whereas for a P-node the number is
multiplied by the factorial of $k$, since in this case all the possible
permutations of the children are allowed.
The same procedure is applied to an M-node; see Section~\ref{sec:mult} for
details.

\begin{algorithm}[!ht]
\begin{algo}
\STATE \Function $N = \pqtreeNperm(T)$
\IF $T$ is a leaf
\label{l1}
\STATE $N =1$
\label{l2}
\ELSE 
	\STATE $\bm{c}=T\ttt{.value}$, $k=\length(\bm{c})$
\label{l5}
	\STATE $p=1$
	\FOR $i = 1,\dots,k$
		\STATE $p = p*\pqtreeNperm(c_i)$
\label{l8}
	\ENDFOR
	\IF $T$ is a Q-node
		\STATE $N = 2*p$
	\ELSE 
		\STATE $N = \factorial(k)*p$
	\ENDIF
\ENDIF	 
\end{algo}
\caption{Compute the number of admissible permutations in a PQ-tree.}
\label{alg:pqtreeNperm}
\end{algorithm}

The toolbox includes an interactive graphical tool for exploring a PQ-tree $T$.
After displaying $T$ by \ttt{pqtreeplot}, it is possible to extract a subtree
by clicking on one node with the left mouse button. In this case, the
corresponding subtree is extracted, it is plotted in a new figure, and it is
saved to the variable \ttt{PQsubtree} in the workspace.
This feature is particularly useful when analyzing a large PQ-tree.
The function \ttt{pqtreeplot} allows to set some attributes of the plot; see
the help page.

\section{A spectral algorithm for the seriation problem}\label{sec:seriation}

In this section we briefly review the spectral algorithm for the seriation
problem introduced in~\cite{atkins1998spectral}, and describe our
implementation.

Given the set of units $U=\{u_1,u_2,\dots,u_n\}$, we will write 
$i\preccurlyeq j$ if $u_i$ precedes $u_j$ in a chosen ordering.
In~\cite{atkins1998spectral}, the authors consider a symmetric bivariate
\emph{correlation function} $f$ reflecting the desire for units $i$ and $j$ to
be close to each other in the sought sequence.
The point is to find all index permutation vectors
$\pib=(\pi_1,\ldots,\pi_n)^T$ such that 
\begin{equation}\label{fperm}
\pi_i\preccurlyeq\pi_j\preccurlyeq\pi_k \quad \iff \quad 
f(\pi_i,\pi_j)\geq f(\pi_i,\pi_k) \quad \text{and} \quad
f(\pi_j,\pi_k)\geq f(\pi_i,\pi_k).
\end{equation}
It is natural to associate to such correlation function a real symmetric matrix
$F$, whose entries are defined by $f_{ij}=f(i,j)$.
This matrix plays exactly the role of the similarity matrix $S$ discussed in
Section~\ref{sec:mathback}, as the following theorem states.

\begin{theorem}\label{theo:fs}
A matrix $F$ is an R-matrix if and only if \eqref{fperm} holds.
\end{theorem}

\begin{proof}
Let us assume that the permutation $\pib$ which realizes \eqref{fperm} has
already been applied to the units. 
Then, since a permutation of the units corresponds to a simultaneous
permutation of the rows and columns of the matrix $F$, we obtain
$$
i\leq j\leq k \quad \iff \quad 
f_{ij}\geq f_{ik} \quad \text{and} \quad
f_{jk}\geq f_{ik}.
$$
The first inequality $f_{ij}\geq f_{ik}$ is exactly \eqref{rmatrix2}.
Keeping into account the symmetry of $F$ and cyclically permuting the indexes,
from the second inequality we get
$$
j\leq k\leq i \quad \iff \quad 
f_{ij}\leq f_{ik},
$$
which corresponds to \eqref{rmatrix1}.
\end{proof}

If a seriation data set is described by an adjacency (or abundance) matrix $A$,
we will set $F=AA^T$.
If $F$ is pre-$R$ (see Section~\ref{sec:mathback}), there exists a rows/columns
permutation that takes it in $R$-form. Unfortunately, this property cannot be
stated in advance, in general. This property can be ascertained, e.g., after
applying the algorithm discussed in this section; see below.

The authors approach in~\cite{atkins1998spectral}, see
also~\cite{estrada2010network}, is to consider the minimization of the
following penalty function 
$$
h(\bm{x}) = \frac{1}{2}\sum_{i,j=1}^{n} f_{ij}(x_i-x_j)^2, 
\quad \bm{x}\in \R^n,
$$ 
whose value is small for a vector $\bm{x}$ such that each pair $(i,j)$ of
highly correlated units is associated to components $x_i$ and $x_j$ with close
values.
Once the minimizing vector $\bm{x}_{\min}$ is computed, it is sorted in
either nonincreasing or nondecreasing value order, yielding
$\bm{x}_\pib=(x_{\pi_1},\ldots,x_{\pi_n})^T$.
The permutation of the units $\pib$ realizes \eqref{fperm}.

Note that $h$ does not have a unique minimizer, since its value does not change
if a constant is added to each of the components $x_i$ of the vector
$\bm{x}$.
In order to ensure uniqueness and to rule out the trivial solution, it is
necessary to impose two suitable constraints on the components of the vector
$\mathbf{x}$.
The resulting minimization problem is:
$$
\begin{aligned}
&\text{minimize} & &
h(\bm{x}) = \frac{1}{2}\sum_{i,j=1}^{n} f_{ij}(x_i-x_j)^2 \\
&\text{subject to} & & \sum_i x_i = 0 \quad \text{and} \quad \sum_i x_i^2 = 1.
\end{aligned}
$$

The solution to this approximated problem may be obtained from the Fiedler
vector of the Laplacian $L$ of the correlation matrix $F$. 
Letting $D = \diag(d_i)$ be the degree matrix, $d_i=\sum_{j=1}^n f_{ij}$, it is
immediate to observe that
$$
h(\bm{x}) = \frac{1}{2}\sum_{i,j=1}^{n} f_{ij}(x_i^2+x_j^2-2x_ix_j)
= \bm{x}^T D \bm{x} - \bm{x}^T F \bm{x}.
$$ 
This shows that the previous minimization problem can be rewritten as 
\begin{eqnarray*}
\min_{\|\mathbf{x}\|=1,\ \bm{x}^T\bm{e} = 0}
\mathbf{x}^TL\mathbf{x}
\end{eqnarray*}
where $L=D-F$.
The constraints require $\mathbf{x}$ to be a unit vector orthogonal to
$\mathbf{e}$.
Being $L$ symmetric, all the eigenvectors except $\bm{e}$ satisfy the
constraints.
Consequently, a Fiedler vector is a solution to the constrained minimization
problem.

In fact, Theorem~3.2 from~\cite{atkins1998spectral} proves that an R-matrix has
a monotone Fiedler vector, while Theorem~3.3, under suitable assumptions,
implies that a reordering of the Fiedler vector takes a pre-R matrix to R-form.
This confirms that the problem is well posed only when $F$ is pre-R.
Nevertheless, real data sets may be inconsistent, in the sense that do not
necessarily lead to pre-R similarity matrices.
In such cases, it may be useful to construct an approximate solution to the
seriation problem, and sorting the entries of the Fiedler vector generates an
ordering that tries to keep highly correlated elements close to each other.
This is relevant because techniques based on Fiedler vectors are being used for
the solution of different sequencing
problems~\cite{barnard1995spectral,greenberg1995physical,higham2007,juvan1992optimal}. 
In particular, they are employed in complex network analysis, e.g., for
community detection and partitioning of graphs~\cite{estrada2012,estrada2015}.

The algorithm proposed in~\cite{atkins1998spectral} is based upon the above
idea, and uses a PQ-tree to store the permutations of the units that produce a
solution to the seriation problem; its implementation is described in
Algorithm~\ref{alg:spectrsort}.

\begin{algorithm}[!ht]
\begin{algo}
%\SetKwInOut{Input}{Input}
%\SetKwInOut{Output}{Output}
%\Input{$F$ an $n\times n$ pre-R matrix,$U$ a set of indices for the rows/columns of $F$, tolerance $\tau$}\\
%\Output{$T$, a PQ-tree}
\STATE \Function $T = \spectrsort(F,U)$
\STATE $n=$ row size of $F$
\STATE $\alpha = \min_{i,j} f_{i,j}$, \If $\alpha\neq 0$, 
	$\bm{e}=(1,\ldots,1)^T$, $F=F-\alpha \bm{e}\bm{e}^{T}$, \End
\label{line2}
\STATE call $\getconcomp$ to construct the connected components 
$\{F_1,\dots,F_k\}$ of $F$
\label{line3}
\STATE \phantom{XXX} and the corresponding index sets $U=\{U_1,\dots,U_k\}$ 
\label{line4}
\IF {$k>1$}
	\FOR {$j=1,\dots,k$}
\label{line6}
		\STATE $v(j) = \spectrsort(F_j,U_j)$
	\ENDFOR
	\STATE $T = \pnode(v)$
\label{line9}
\ELSE 
	\IF {$n=1$}
\label{line11}
		\STATE $T=\lnode(U)$
	\ELSEIF {$n=2$}
		\STATE $T=\pnode(U)$
	\ELSE
\label{line15}
		\STATE $L=$ Laplacian matrix of $F$
\label{line16}
		\STATE compute (part of) the eigenvalues and eigenvectors of $L$
\label{line17}
		\STATE determine multiplicity $n_F$ of the Fiedler value
			according to a tolerance $\tau$
		\IF {$n_F = 1$}
			\STATE $\bm{x}=$ sorted Fiedler vector
			\STATE $t$ number of distinct values in $\bm{x}$
				according to a tolerance $\tau$
\label{line19}
			\FOR {$j=1,\dots,t$}
				\STATE $u_j$ indices of elements in $\bm{x}$ 
					with value $x_j$
				\IF $u_j$ has just one element
					\STATE $v_j=\lnode(u_j)$
\label{line26}
				\ELSE 
					\STATE $v(j) = \spectrsort(F(u_j,u_j),U(u_j,u_j))$
\label{line28}
				\ENDIF
			\ENDFOR
			\STATE $T = \qnode(v)$
		\ELSE
			\STATE $T=\mnode(U)$
\label{line33}
		\ENDIF
	\ENDIF
\ENDIF
\end{algo}
\caption{Spectral sort algorithm.}
\label{alg:spectrsort}
\end{algorithm}

The algorithm starts translating all the entries of the correlation matrix so
that the smallest is 0 , i.e., 
\begin{equation}\label{trans}
\tilde{F}=F-\alpha \bm{e}\bm{e}^{T}, \qquad \alpha = \min_{i,j} f_{ij};
\end{equation}
see line~\ref{line2} of Algorithm~\ref{alg:spectrsort}.
This is justified by the fact that $F$ and $\tilde{F}$ have the same Fiedler
vectors and that if $F$ is an irreducible R-matrix such translation ensures
that the Fiedler value is a simple eigenvalue of
$L$~\cite[Lemma~4.1~and~Theorem~4.6]{atkins1998spectral}.
Our software allows the user to disable this procedure (see
Table~\ref{tab:opts} below)
as he may decide to
suitably preprocess the similarity matrix in order to reduce the
computational load.
Indeed, the translation procedure is repeated each time the algorithm calls
itself recursively.

\begin{algorithm}[!ht]
\begin{algo}
\STATE \Function $U = \getconcomp(F)$
\STATE preallocate the cell-array $U$, $chlist=$empty vector
\STATE $root=$\{node 1\}, $list=root$, $n=$row size of $F$
\STATE $i=0$, $flag=true$ (logical variable)
\WHILE $flag$
\STATE $i =i+1$
\STATE $list =\graphvisit(root,list)$
\STATE $U\{i\}=list$
\STATE update $chlist$ adding the nodes in $list$ and sort the vector
\STATE $flag=true$ if the number of elements in $chlist$ is different from $n$
\STATE \phantom{XXX} otherwise $flag=false$
\IF $flag$
	\STATE choose the $root$ for a new connected component
	\IF there are no connected components left
		\STATE exit
	\ENDIF
	\STATE $list=root$
\ENDIF
\ENDWHILE
\end{algo}
\caption{Detect the connected components of a graph.}
\label{alg:getconcomp}
\end{algorithm}

If the matrix $F$ is reducible, then the seriation problem can be
decoupled~\cite[Lemma~4.2]{atkins1998spectral}.
Lines~\ref{line3}--\ref{line4} of the algorithm detect the irreducible blocks
of the correlation matrix by using the function \textsf{getconcomp.m}, which
also identifies the corresponding index sets.
The function, described in Algorithm~\ref{alg:getconcomp}, constructs a
\emph{cell array} containing the indices which identify each connected
component of a graph.
It calls the function \textsf{graphvisit.m}, which visits a graph starting from
a chosen node; see Algorithm~\ref{alg:graphvisit}.
Note that these two functions, in order to reduce the stack consumption due to
recursion, use a global variable to store the correlation matrix.

\begin{algorithm}[!ht]
\begin{algo}
\STATE \Function $list = \graphvisit(root,list)$
\STATE construct the list $l$ of the indices of the nodes connected to the root
\STATE initialize an empty list $nlist$
\STATE find the elements of $l$ which are not in $list$
\STATE add the new elements to $list$ and to $nlist$
\IF $nlist$ is not empty
	\STATE sort $list$
	\FOR each node $i$ in $nlist$
		\STATE $list=\graphvisit(nlist(i),list)$
	\ENDFOR
\ENDIF
\end{algo}
\caption{Visit a graph starting from a node.}
\label{alg:graphvisit}
\end{algorithm}

If more than one connected component is found, then the function calls itself
on each component, and stores the returned lists of nodes as children of
a P-node (lines \ref{line6}--\ref{line9}).
If the matrix is irreducible, the dimension $n$ of the matrix is
considered (lines~\ref{line11}--\ref{line15}).
The cases $n=1,2$ are trivial.
If $n>2$, the Laplacian matrix $L$ is computed, as well as the Fiedler value
and vector (lines~\ref{line16}--\ref{line17}). 
Depending on the matrix being ``small'' or ``large'' different algorithms are
used.
For a small scale problem the full spectral decomposition of the Laplacian is
computed by the \ttt{eig} function of Matlab.
For a large scale problem only a small subset of the eigenvalues and
eigenvectors are evaluated using the \ttt{eigs} function, which is based on a
Krylov space projection method.
The \pqser toolbox computes by default the eigenpairs corresponding to the
three eigenvalues of smallest magnitude, since they are sufficient to
understand if the Fiedler value is simple or multiple, but the default value
can be modified.
The choice between the two approaches is automatically performed and it may be
influenced by the user; see Table~\ref{tab:opts} in Section~\ref{sec:implser}.

Then, the algorithm determines the multiplicity of the Fiedler value according
to a given tolerance.
If the Fiedler value is a simple eigenvalue of $L$, the algorithm sorts the
elements of the current list according to the reordering of the Fiedler vector
and stores them as the children of a Q-node.
If two or more values of the Fiedler vector are repeated the function invokes
itself recursively (line~\ref{line28}), in accordance
with~\cite[Theorem~4.7]{atkins1998spectral}; on the contrary, the corresponding
node becomes a leaf (line~\ref{line26}).
In our implementation we introduce a tolerance $\tau$ to distinguish ``equal''
and ``different'' numbers: $a$ and $b$ are considered ``equal'' if
$|a-b|<\tau$. The default value for $\tau$ is $10^{-8}$.

In the case of a multiple Fiedler value, the algorithm conventionally
constructs an ``M-node'' (line~\ref{line33}).
This new type of node has been introduced in order to flag this particular
situation, which will be further discussed in Section~\ref{sec:mult}.

Algorithm~\ref{alg:spectrsort} produces a PQ-tree whether $F$ is a pre-R matrix
or not.
If all the Fiedler values computed are simple, the starting matrix is pre-R and
any permutation encoded in the PQ-tree will take it to R-form.
In the presence of a multiple Fiedler vector the problem is not well posed and
an approximate solution is computed.

The number $N$ of all the admissible permutations generated by the algorithm
can be obtained by counting all the admissible boundaries of the tree.
In the case of a PQ-tree consisting of a single Q-node $N$ is equal to 2,
because only the left-to-right order of the children leaves and its reverse
are possible. 
For a single P-node, the number of all the permutations is the factorial of the
number of the children.
An M-node is temporarily treated as a P-node, although we experimentally
observed that not all the permutations are admissible; this aspect is discussed
in Section~\ref{sec:mult}.

\subsection{Implementation of spectral seriation}\label{sec:implser}

The functions included in the \pqser toolbox are listed in
Table~\ref{tab:serfuncs}.
Besides the function \ttt{spectrsort}, which implements 
Algorithm~\ref{alg:spectrsort}, there is a parallel version of the same method,
\ttt{pspectrsort}, which distributes the \ttt{for} loop at
line~\ref{line6} of the algorithm among the available processing units.
In order to execute the function \ttt{pspectrsort}, the Parallel Computing
Toolbox must be present in the current Matlab installation.

\begin{table}[htb]
\footnotesize
\centering
\begin{tabular}{ll}
\hline
\ttt{spectrsort}   & spectral sort for the seriation problem \\
\ttt{pspectrsort}  & parallel version of \ttt{spectrsort} \\
\ttt{fiedvecs}     & compute the Fiedler vectors and values of a Laplacian \\
\ttt{getconcomp}   & determine the connected components of a graph \\
\ttt{graphvisit}   & visit a graph starting from a node \\
\ttt{distinct}     & sort and level the elements of a vector \\
\ttt{lapl}         & construct the graph Laplacian of a matrix \\
\ttt{testmatr}     & test matrices for PQser \\
\hline
\end{tabular}
\caption{Functions in the \pqser toolbox devoted to the solution of the
seriation problem.}
\label{tab:serfuncs}
\end{table}

The function \ttt{testmatr} allows one to create some simple test problems.
The remaining functions of Table~\ref{tab:serfuncs} are not likely to be used
in the common use of the toolbox. 
They are made available to the expert user, who may decide to call them
directly or to modify their content.

\begin{table}[htb]
\footnotesize
\centering
\begin{tabular}{ll}
\hline
\ttt{tau} & tolerance used to distinguish between ``equal'' and ``different'' \\
	  & values (\ttt{spectrsort} and \ttt{fiedvecs}, def. $10^{-8}$) \\
\ttt{translate} & applies translation \eqref{trans} (\ttt{spectrsort}, def. 1) \\
\ttt{lrg} & used to select small scale or large scale algorithm (\ttt{fiedvecs}, \\
	& true if the input matrix is sparse) \\
\ttt{nlarge} & if matrix size is below this value, the small scale algorithm \\
	     & is used (\ttt{fiedvecs}, def. 1000) \\
\ttt{neig} & number of eigenpairs to be computed when the large scale \\
	   & algorithm is used (\ttt{fiedvecs}, def. 3) \\
\ttt{maxncomp} & maximum number of connected components (\ttt{getconcomp}, def. 100) \\
\ttt{bw} & half bandwidth of test matrix (\ttt{testmatr}, type 2 example, def. 2) \\
\ttt{spar} & construct a sparse test matrix (\ttt{testmatr}, type 2 example, def. 1) \\
\hline
\end{tabular}
\caption{Tuning parameters for the \pqser toolbox; the functions affected are
reported in parentheses, together with the default value of each parameter.}
\label{tab:opts}
\end{table}

The toolbox has some tuning parameters, which are set to a default value, but
can be modified by the user. This can be done by passing to a function, as an
optional argument, a variable of type \ttt{struct} with fields chosen among the
ones listed in Table~\ref{tab:opts}.
For example:
\begin{quote}
\footnotesize
\begin{verbatim}
opts.translate = 0;
T = spectrsort(F,opts);
\end{verbatim}
\end{quote}
applies Algorithm~\ref{alg:spectrsort} to a similarity matrix $F$ omitting
the translation process described in \eqref{trans}.

To illustrate the use of the toolbox, we consider a similarity matrix $R$
satisfying the Robinson criterion 
$$
R = \begin{bmatrix}
200 & 150 & 120 & 80 & 40 & 0 & 0 & 0 & 0 & 0\\
150 & 200 & 160 & 120 & 80 & 40 & 0 & 0 & 0 & 0\\
120 & 160 & 200 & 160 & 120 & 80 & 40 & 0 & 0 & 0\\
80 & 120 & 160 & 200 & 160 & 120 & 80 & 40 & 0 & 0\\
40 & 80 & 120 & 160 & 200 & 160 & 120 & 80 & 40 & 0\\
0 & 40 & 80 & 120 & 160 & 200 & 160 & 120 & 80 & 40\\
0 & 0 & 40 & 80 & 120 & 160 & 200 & 160 & 120 & 80\\
0 & 0 & 0 & 40 & 80 & 120 & 160 & 200 & 160 & 120\\
0 & 0 & 0 & 0 & 40 & 80 & 120 & 160 & 200 & 150\\
0 & 0 & 0 & 0 & 0 & 40 & 80 & 120 & 150 & 200
\end{bmatrix}
$$
and the pre-R matrix obtained by applying to the rows and columns of $R$ a
random permutation 
$$
F = \begin{bmatrix} 
200 & 0 & 0 & 150 & 120 & 0 & 160 & 40 & 0 & 80\\
0 & 200 & 150 & 0 & 0 & 120 & 0 & 80 & 160 & 40\\ 
0 & 150 & 200 & 0 & 0 & 80 & 0 & 40 & 120 & 0\\ 
150 & 0 & 0 & 200 & 80 & 0 & 120 & 0 & 0 & 40\\ 
120 & 0 & 0 & 80 & 200 & 80 & 160 & 120 & 40 & 160\\ 
0 & 120 & 80 & 0 & 80 & 200 & 40 & 160 & 160 & 120\\ 
160 & 0 & 0 & 120 & 160 & 40 & 200 & 80 & 0 & 120\\ 
40 & 80 & 40 & 0 & 120 & 160 & 80 & 200 & 120 & 160\\ 
0 & 160 & 120 & 0 & 40 & 160 & 0 & 120 & 200 & 80\\ 
80 & 40 & 0 & 40 & 160 & 120 & 120 & 160 & 80 & 200 
\end{bmatrix}.
$$

The PQ-tree $T$ containing the solution of the reordering problem is
constructed by calling the function \ttt{spectrsort}, which returns the
resulting data structure:
\begin{quote}
\footnotesize
\begin{verbatim}
T = spectrsort(F,opts)
T = 
  struct with fields:
     type: 'Q'
    value: [1x10 struct]
\end{verbatim}
\end{quote}
Using the function \ttt{pqtreeplot} 
\begin{quote}
\footnotesize
\begin{verbatim}
pqtreeplot(T)
\end{verbatim}
\end{quote}
we obtain the representation of the PQ-tree displayed in Figure~\ref{pqtree2}.

\begin{figure}
\begin{center}
\includegraphics[width=.52\textwidth]{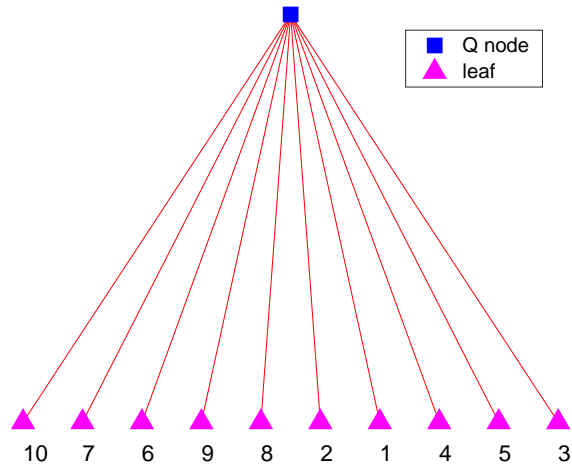}
\caption{A PQ-tree corresponding to a pre-R matrix of dimension 10.}
\label{pqtree2}
\end{center}
\end{figure}

In this particular case, the PQ-tree $T$ consists of just a Q-node as a root,
so only two permutations of the leaves are allowed.
They can be extracted from the tree using the function \ttt{pqtreeperms}, whose
output is
\begin{quote}
\footnotesize
\begin{verbatim}
perms_matrix = pqtreeperms(T)
perms_matrix =
     4     1     7     5    10     8     6     9     2     3
     3     2     9     6     8    10     5     7     1     4
\end{verbatim}
\end{quote}   
In some occasions, a PQ-tree may contain a very large number of permutations.
In such cases, the function \ttt{pqtree1perm} extracts just one of the
permutations, in order to apply it to the rows and columns of the matrix $F$:
\begin{quote}
\footnotesize
\begin{verbatim}
seq = pqtree1perm(T);
AR = F(seq,seq);
\end{verbatim}
\end{quote}   
Since $F$ is pre-R, we clearly reconstruct the starting similarity matrix $R$.

This experiment is contained in the script \texttt{graf2.m}, in the
\texttt{demo} sub-directory.
The script \texttt{tutorial.m}, in the same directory, illustrates the use of
the toolbox on other numerical examples.

\section{The case of a multiple Fiedler value}\label{sec:mult}

In this section we discuss the case where the Fiedler value is a multiple root
of the characteristic polynomial of the Laplacian $L$.
When this happens, the eigenspace corresponding to the smallest nonzero
eigenvalue of $L$ has dimension larger than one, so there is no uniqueness in
the choice of the Fiedler vector.

We conjecture that sorting the entries of a Fiedler vector, that is, of any
vector in the eigenspace of the Fiedler value, does not necessarily lead to all
possible indexes permutations, i.e., the factorial of the number $n$ of units.
We observed that there may be some constraints that limit the number of
permutations deriving from the Fiedler vector, and this number does not appear
to be related to the multiplicity of the Fiedler value by a simple formula.
We will illustrate this issue by a numerical experiment.

As we did not find any reference to this problem in the literature, we plan
to study it in a subsequent paper.
This is the reason why the \pqser toolbox conventionally associates an
\emph{M-node}, to the presence of a multiple Fiedler value.
In the software, the new type of node is temporarily treated as a P-node.
This leaves the possibility to implement the correct treatment of the case,
once the problem will be understood.

\begin{figure}[hbt]
\begin{center}
\begin{tikzpicture}[->,>=stealth',auto,node distance=1cm,
	on grid,semithick, every state/.style={fill=cyan!20!white,draw=none,
	circular drop shadow,text=black,inner sep=0pt}]
	\tiny
\node[state](A){1};
\node[state](B)[below=of A]{2};
\node[state](C)[below=of B]{3};
\node[state](D)[below=of C]{4};
\node[state](E)[below=of D]{5};
\node[state,fill=red!20!white,node distance=4cm](L)[right=of A]{1};
\node[state,fill=red!20!white,node distance=4cm](M)[right=of B]{2};
\node[state,fill=red!20!white,node distance=4cm](N)[right=of C]{3};
\node[state,fill=red!20!white,node distance=4cm](O)[right=of D]{4};
\node[state,fill=red!20!white,node distance=4cm](P)[right=of E]{5};
\path (A) edge (L);
\path (A) edge (M);
\path (B) edge (M);
\path (B) edge (N);
\path (C) edge (N);
\path (C) edge (O);
\path (D) edge (O);
\path (D) edge (P);
\path (E) edge (P);
\path (E) edge (L);
\end{tikzpicture}
\caption{The \emph{cycle} seriation problem; the units are on the left, the
types on the right.}
\label{fig:cycle}
\end{center}
\end{figure}
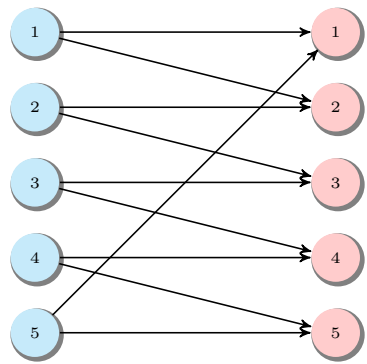

Here we present a simple example to justify our conjecture.
Let us consider the seriation problem described by the bipartite graph depicted
in Figure~\ref{fig:cycle}.
The nodes on the left are the units, e.g., the excavation sites; on the right
there are the types, which may be seen as the archeological findings. The
relationships between units and types are represented by edges connecting the
nodes.

The problem is clearly unsolvable, as the associated graph describes a
\emph{cycle}: each unit is related to surrounding units by a connection to a
common type, and the two extremal units are related to each other in the same
way.

At the same time, not all the units permutations are admissible.
For example, one may argue that the permutation $\pib_1=(3,4,5,1,2)^T$ should
be considered partially feasible, as it breaks only one of the constraints
contained in the bipartite graph, while the ordering $\pib_2=(1,4,2,5,3)^T$ has
nothing to do with the problem considered.

The adjacency matrix associated to the graph in Figure~\ref{fig:cycle} is the
following
\begin{equation}\label{incidCycle}
A = \begin{bmatrix}
1 & 1 & 0 & 0 & 0\\
0 & 1 & 1 & 0 & 0\\
0 & 0 & 1 & 1 & 0\\
0 & 0 & 0 & 1 & 1\\
1 & 0 & 0 & 0 & 1
\end{bmatrix}.
\end{equation}
We can associate to \eqref{incidCycle} the similarity matrix
\begin{equation}\label{cycleF}
F = AA^T = \begin{bmatrix}
2 & 1 & 0 & 0 & 1\\
1 & 2 & 1 & 0 & 0\\
0 & 1 & 2 & 1 & 0\\
0 & 0 & 1 & 2 & 1\\
1 & 0 & 0 & 1 & 2
\end{bmatrix},
\end{equation}
whose Laplacian is
$$
L = D - F = \begin{bmatrix}
2 & -1 & 0 & 0 & -1\\
-1 & 2 & -1 & 0 & 0\\
0 & -1 & 2 & -1 & 0\\
0 & 0 & -1 & 2 & -1\\
-1 & 0 & 0 & -1 & 2
\end{bmatrix}.
$$

The matrix $L$ is circulant, that is, it is fully specified by its first
column, while the other columns are cyclic permutations of the first one
with an offset equal to the column index. 
A complete treatment of circulant matrices can be found
in~\cite{davis1979circulant}, while~\cite{redivo2012smt} implements a Matlab
class for optimized circulant matrix computations.

One of the basic properties of circulant matrices is that their spectrum is
analytically known. In particular, the eigenvalues of $L$ are given by
$$
\{\widehat{L}(1), \widehat{L}(\omega), \widehat{L}(\omega^2),
\widehat{L}(\omega^3), \widehat{L}(\omega^4) \},
$$
where $\widehat{L}(\zeta)$ is the discrete Fourier transform of the first
column of $L$
$$
\hat{L}(\zeta) = 2 -\zeta ^{-1}-\zeta ^{-4},
$$
and $\omega = \e^{\frac{2\pi\ii}{5}}$ is the minimal phase $5$th root of
unity; see~\cite{davis1979circulant}.
A simple computation shows that
$$
\widehat{L}(1)=0, \quad
\widehat{L}(\omega)=\widehat{L}(\omega^4)=2-2\cos\frac{2\pi}{5}, \quad
\widehat{L}(\omega^2)=\widehat{L}(\omega^3)=2-2\cos\frac{4\pi}{5}, 
$$
so that the Fiedler value $\widehat{L}(\omega)$ has multiplicity 2.

To explore this situation we performed the following numerical experiment.
We considered 10000 random linear combinations of an orthonormal basis for the
eigenspace corresponding to the Fiedler value. This produces a set of random
vectors, belonging to a plane immersed in $\R^5$, which can all be considered
as legitimate ``Fiedler vectors''.

Each vector was sorted, and the corresponding permutations of indexes were
stored in the columns of a matrix.
In the end, all the repeated permutations were removed.
We obtained 10 permutations, reported in the columns of the following matrix
\begin{equation}\label{permat}
\begin{bmatrix}
3 & 2 & 5 & 3 & 2 & 4 & 4 & 1 & 5 & 1 \\
2 & 1 & 4 & 4 & 3 & 3 & 5 & 5 & 1 & 2 \\
4 & 3 & 1 & 2 & 1 & 5 & 3 & 2 & 4 & 5 \\
1 & 5 & 3 & 5 & 4 & 2 & 1 & 4 & 2 & 3 \\
5 & 4 & 2 & 1 & 5 & 1 & 2 & 3 & 3 & 4
\end{bmatrix}.
\end{equation}
They are much less than the $5!=120$ possible permutations, and reduce to 5 if
we remove the columns which are the reverse of another column.
This confirms our conjecture: when a Fiedler value is multiple some constraints
are imposed on the admissible permutations of the units.

It is relevant to notice that matrix \eqref{permat} does not contain the cyclic
permutations of the units depicted in Figure~\ref{fig:cycle}.
Indeed, the spectral algorithm aims at moving the nonzero components close to
the main diagonal, and this contrasts with the presence of nonzeros in the
elements $f_{51}$ and $f_{15}$ of matrix \eqref{cycleF}.
All permutations contained in \eqref{permat}, when applied to the similarity
matrix $F$, produce the same matrix
$$
\tilde{F} = \begin{bmatrix}
2 & 1 & 1 & 0 & 0 \\
1 & 2 & 0 & 1 & 0 \\
1 & 0 & 2 & 0 & 1 \\
0 & 1 & 0 & 2 & 1 \\
0 & 0 & 1 & 1 & 2
\end{bmatrix},
$$
which exhibits a smaller bandwidth than \eqref{cycleF}.

This experiment can be reproduced by executing the script \texttt{mfiedval.m},
which is found in the \texttt{demo} sub-directory.

%%%%%%%%%%%%%%%%%%%%%%%%%%%%%%%%%%%%%%%%%%%%%%%%%%%%%%%%%%%%%%%%%%%%%
\section{Numerical experiments}\label{sec:numexp}

In this section we illustrate the application of the \pqser toolbox to some
numerical examples.
The experiments can be repeated by running the related Matlab scripts located
in the \texttt{demo} sub-directory of the toolbox.

\begin{figure}[hbt]
\begin{center}
\includegraphics[width=.30\textwidth]{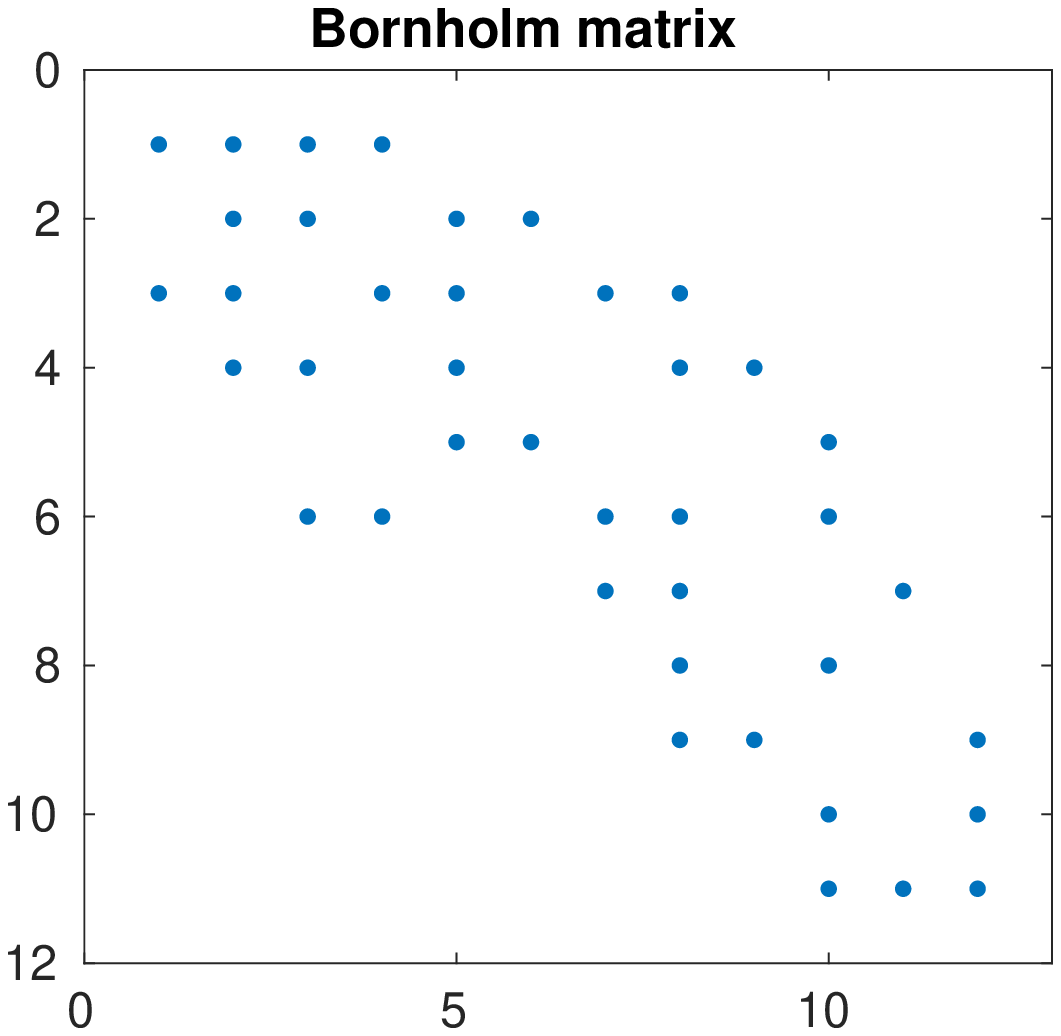}\hfil 
\includegraphics[width=.30\textwidth]{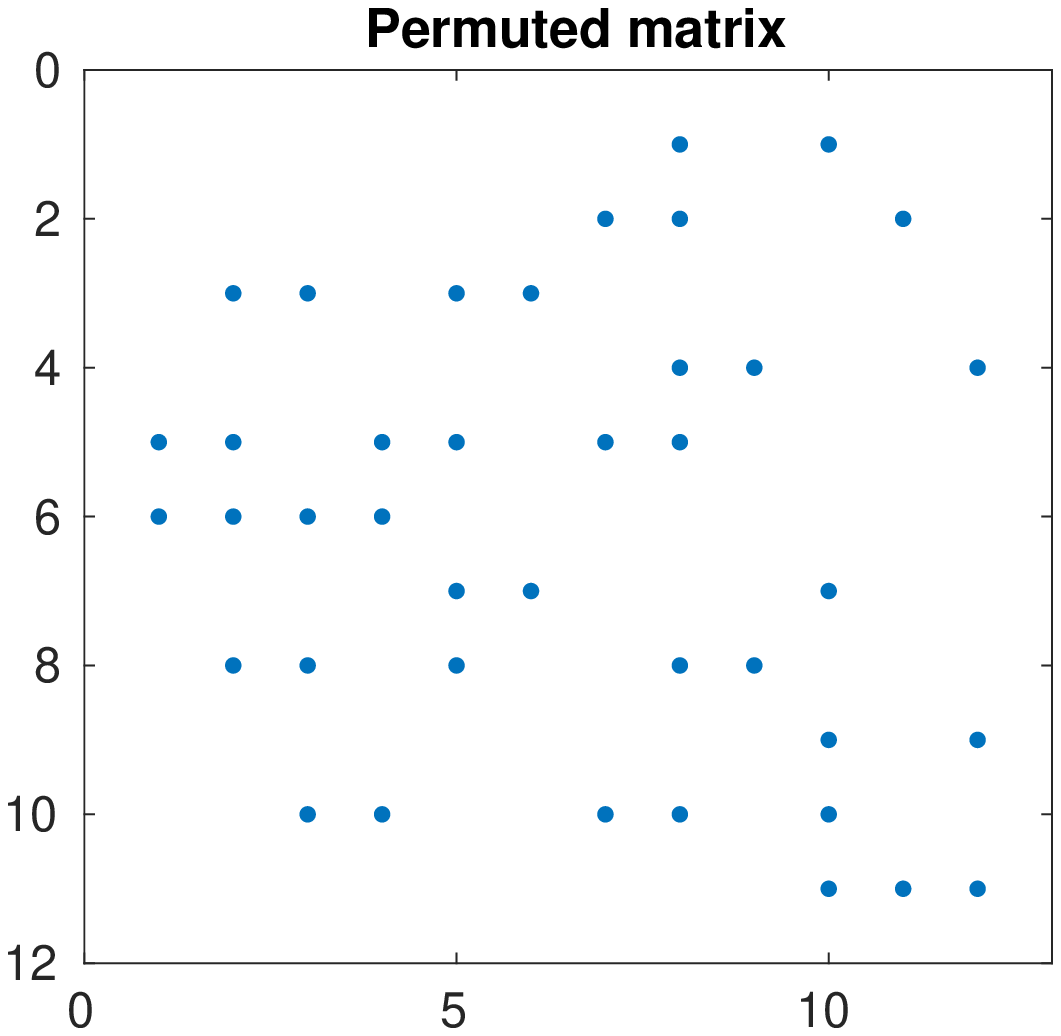}\hfil 
\includegraphics[width=.30\textwidth]{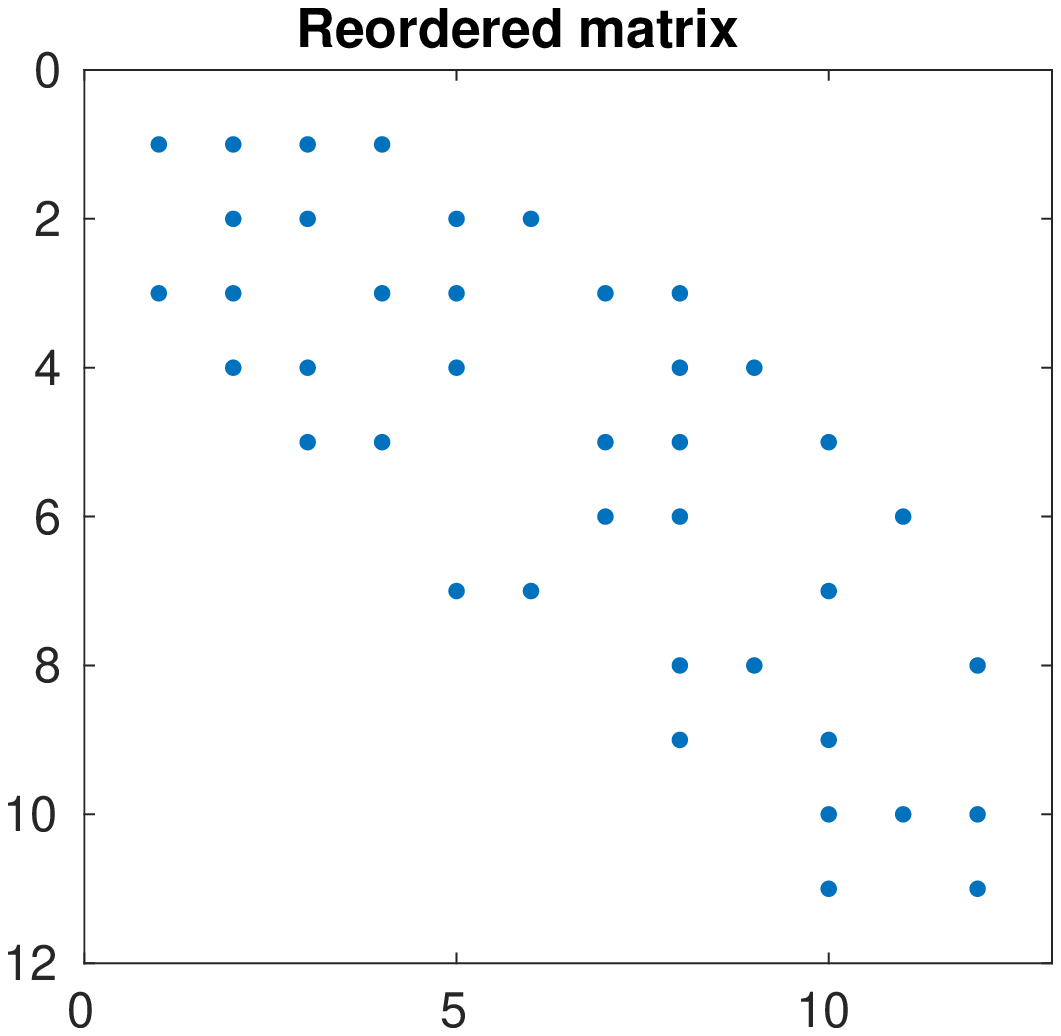}
\caption{Processing of the Bornholm data set: the spy plot on the left shows
the original matrix, a permuted version is reported in the central graph, on
the right we display the matrix reordered by the spectral algorithm.}
\label{bornholm}
\end{center}
\end{figure}

The first example is the numerical processing of the Bornholm data set,
presented in Table~\ref{tab:adjacency}.
We randomly permute the rows of the adjacency matrix and apply the
spectral algorithm to the similarity matrix associated to the permuted matrix. 
The resulting PQ-tree contains just a Q-node, so there is only one solution
(actually, this is a proof that the matrix is pre-R) which we use to reorder
the permuted matrix.
The computational code is contained in the file \ttt{exper1.m}.

Figure~\ref{bornholm} reports the spy plots which represent the nonzero
entries of the initial matrix, its permuted version, and the final reordering.
It is immediate to observe that the lower band of the reordered matrix is
slightly narrower than the initial matrix, showing that the spectral algorithm
was able to improve the results obtained empirically by archaeologists.

The second example concerns the comparison of the spectral algorithm with its
parallel version in the solution of a large scale problem.
The experiments were performed on a dual Xeon CPU E5-2620 system (12 cores),
running the Debian GNU/Linux operating system and Matlab 9.2.

The function \ttt{testmatr} of the toolbox allows the user to create a block
diagonal matrix, formed by $m$ banded blocks whose size is chosen using a
second input parameter.
The matrix is randomly permuted in order to hide its reducible structure.

\begin{figure}[hbt]
\begin{center}
\includegraphics[width=.49\textwidth]{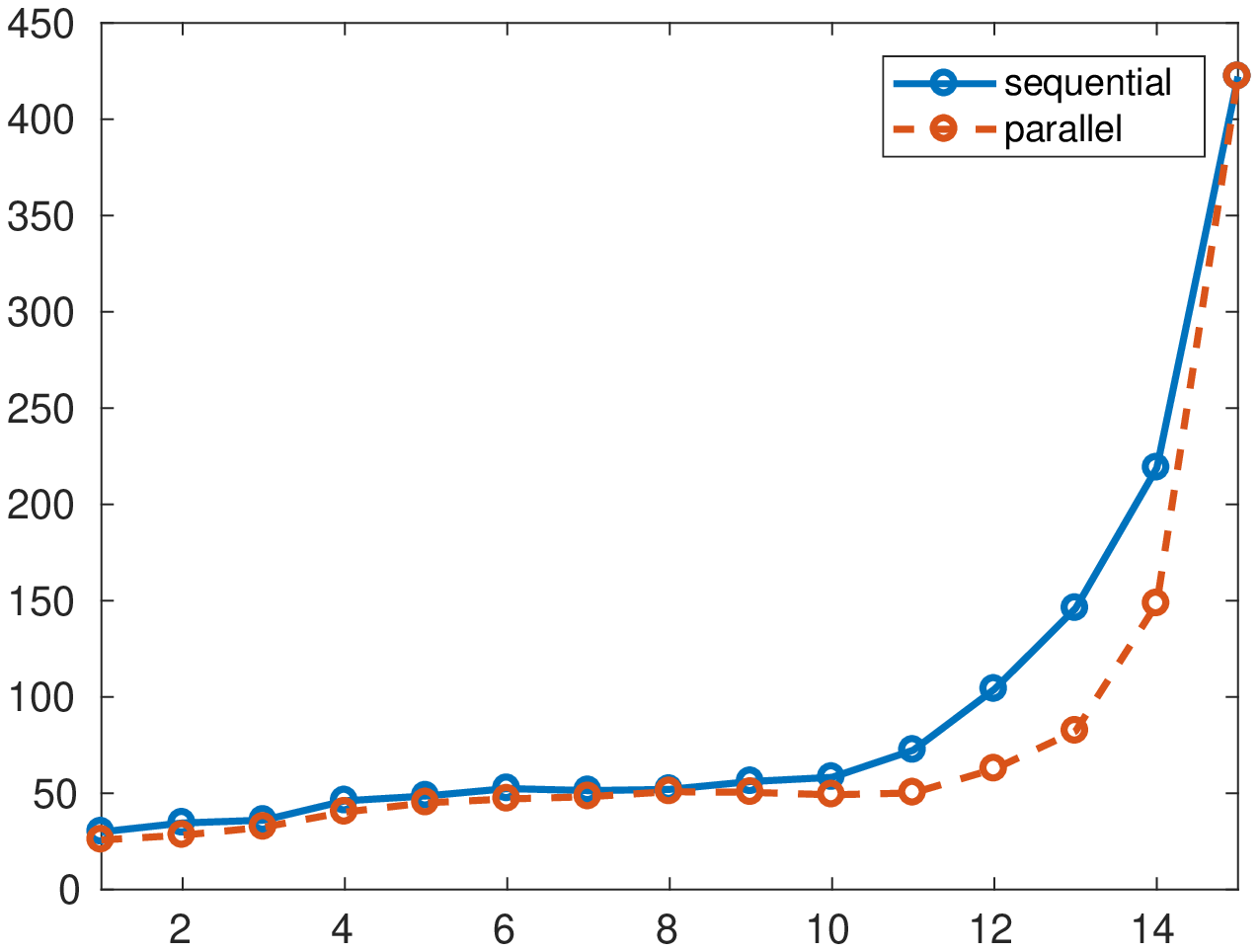}\hfil 
\includegraphics[width=.47\textwidth]{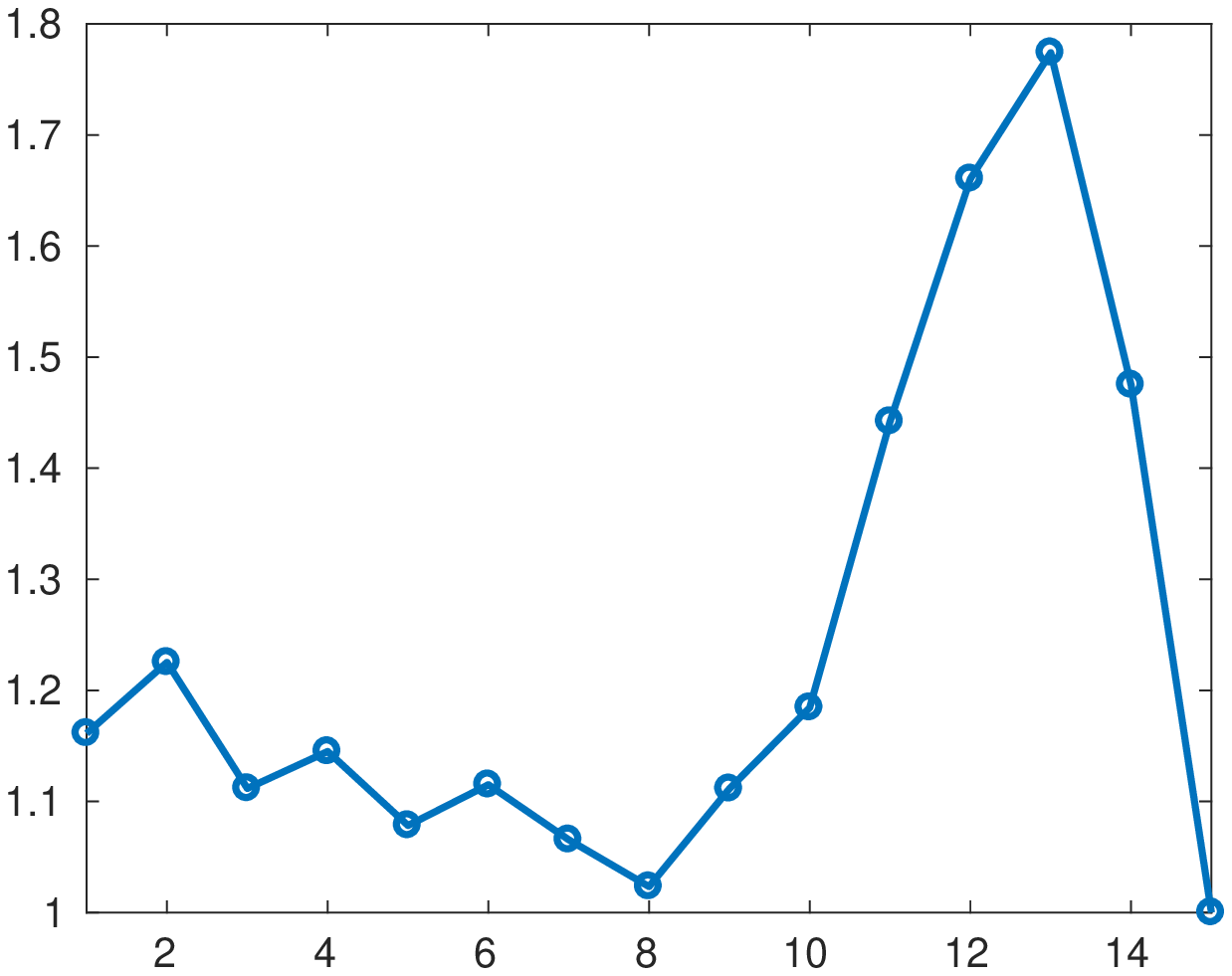}
\caption{Comparison between the sequential and the parallel versions of
Algorithm~\ref{alg:spectrsort}: on the left the execution time in seconds, on
the right the parallel speedup, defined as the ratio between the sequential and
the parallel timings. The test matrix is of dimension $2^{15}=32768$, the size
of each reducible block is $2^j$, where $j$ is  the index reported in the
horizontal axis.}
\label{exper2}
\end{center}
\end{figure}

We let the size of the problem be $n=2^{15}=32768$ and, for $j=1,2\ldots,15$,
we generate a sequence of test matrices containing $n\cdot 2^{-j}$ blocks, each
of size $2^j$.

We apply the function \ttt{spectrsort} that implements
Algorithm~\ref{alg:spectrsort} to the above problems, as well as its parallel
version \ttt{pspectrsort}, and record the execution time; see the file
\ttt{exper2.m}.
The number of processors available on our computer was 12.

The graph on the left of Figure~\ref{exper2} shows that there is a significant
advantage from running the toolbox on a parallel computing system when the
network associated to the problem is composed by a small number of large
connected components.
This is confirmed by the plot of the parallel speedup, that is, the ratio
between the timings of the sequential and the parallel implementations,
displayed in the graph on the right in the same figure.

\begin{figure}[hbt]
\begin{center}
\includegraphics[width=.30\textwidth]{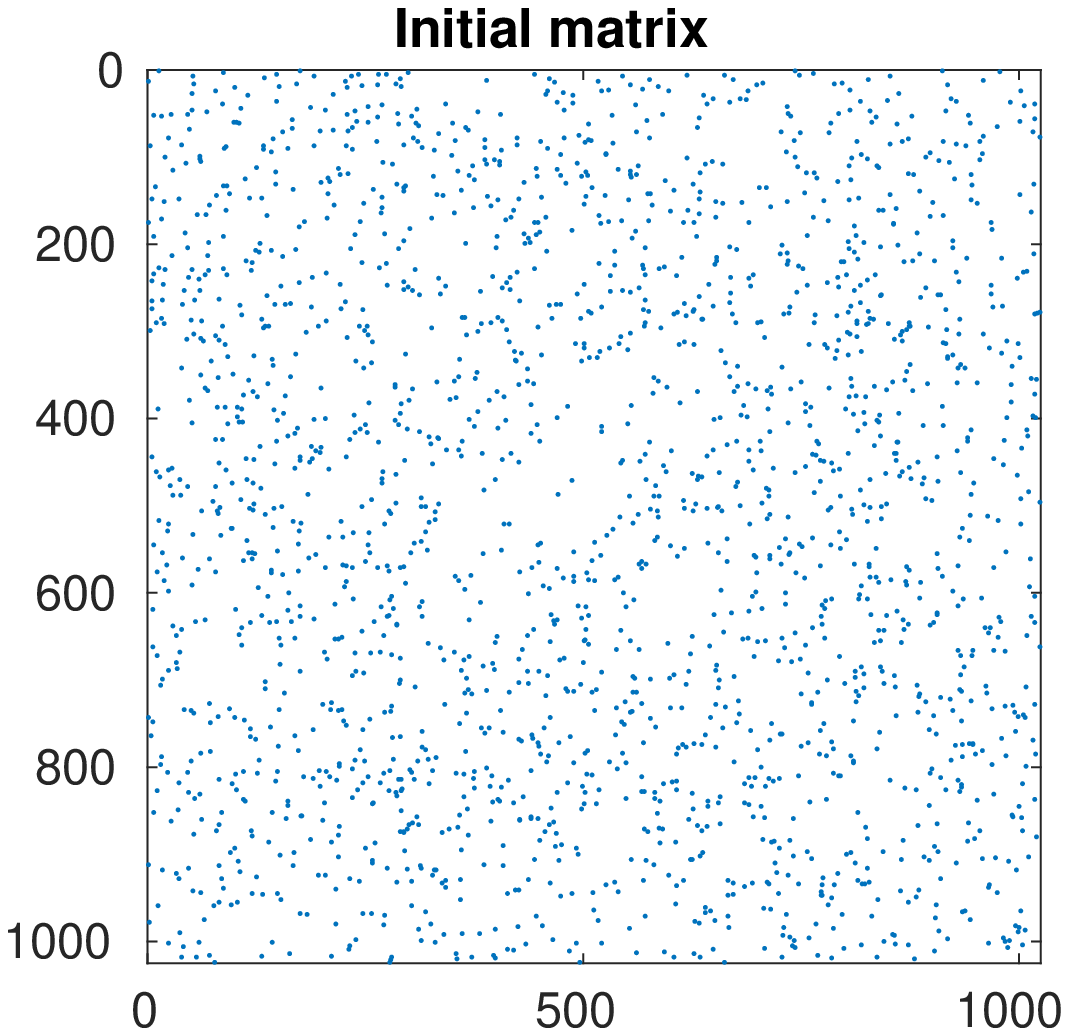}\hfil 
\includegraphics[width=.30\textwidth]{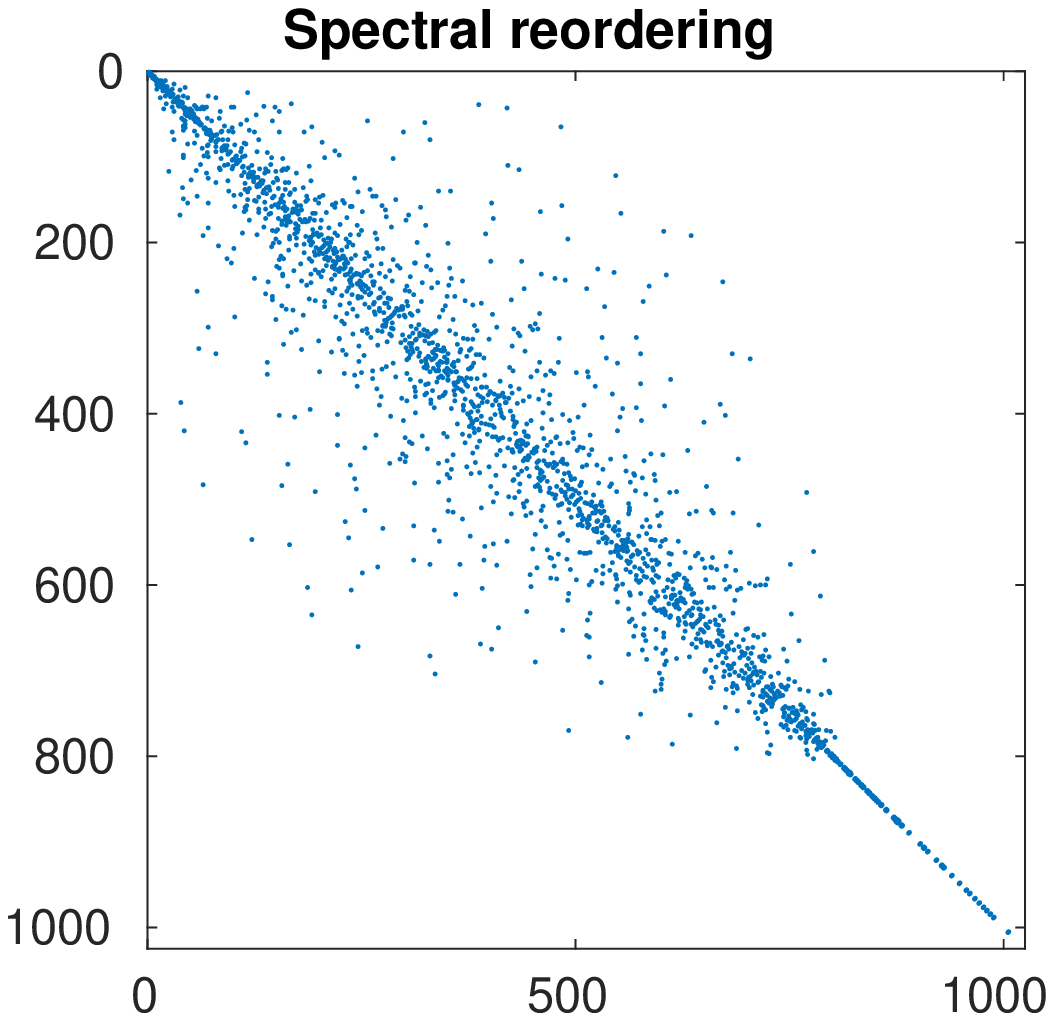}\hfil 
\includegraphics[width=.30\textwidth]{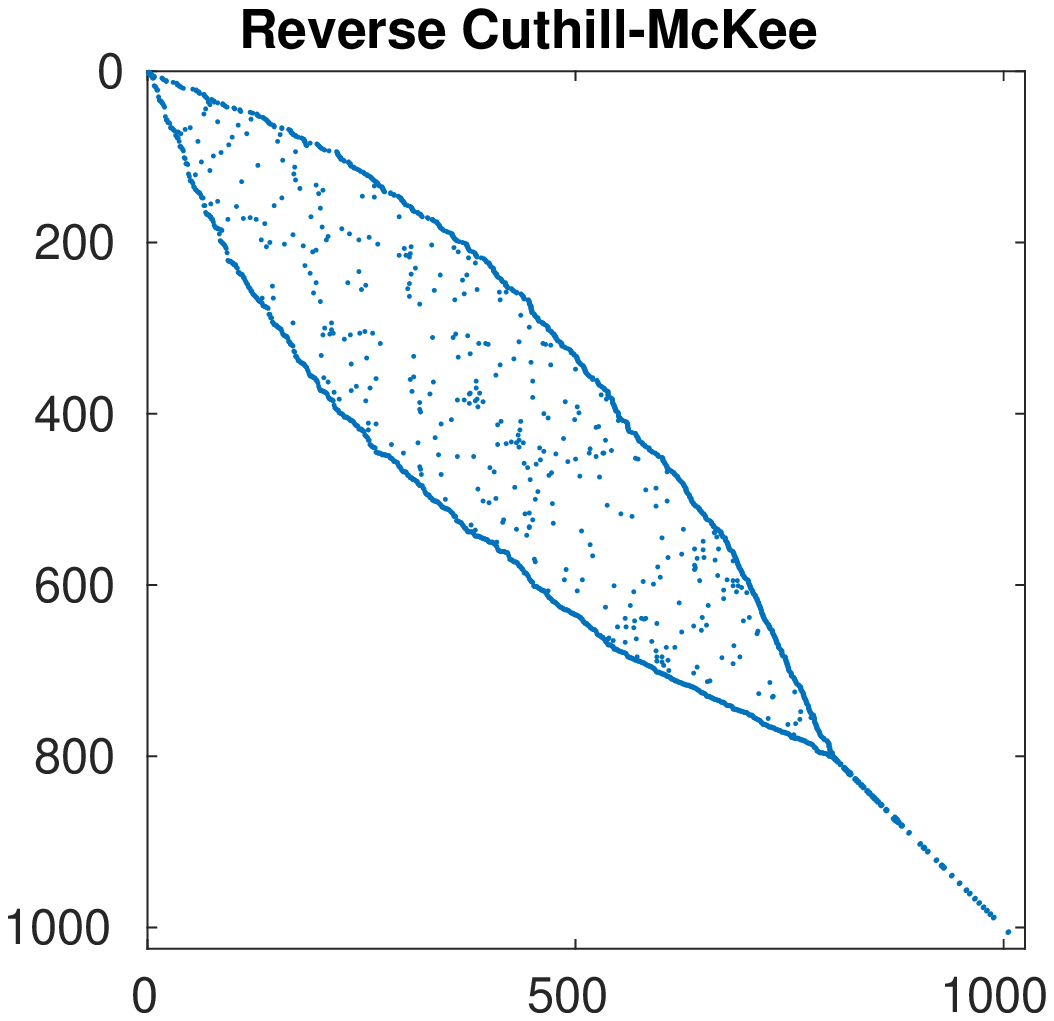}
\caption{Bandwidth reduction of a sparse matrix of size 1024: the three spy
plots display the initial matrix, the reordered matrix resulting from the
spectral algorithm, and the one produced by the \ttt{symrcm} function of
Matlab.}
\label{exper3}
\end{center}
\end{figure}

To conclude, we consider another important application of the reordering
defined by the Fiedler vector of the Laplacian, namely, the reduction of the
bandwidth for a sparse matrix; see~\cite{barnard1995spectral}.

We generate a sparse symmetric matrix of size $n=1024$, having approximately
0.2\% nonzero elements, and reorder its rows and columns by
Algorithm~\ref{alg:spectrsort}.
Notice that in this case the spectral algorithm must be applied to a matrix
whose elements are taken in absolute value.
The computation is described in the script \ttt{exper3.m}.

The resulting matrix is depicted by displaying its nonzero pattern in
Figure~\ref{exper3}, where it is compared to the reverse Cuthill-McKee
ordering, as implemented in the \ttt{symrcm} function of Matlab.
The spectral algorithm appears to be less effective than \ttt{symrcm},
leading to a reordered matrix with a wider band.
This is due to the fact that \ttt{spectrsort} aims at placing the largest
entries close to the diagonal, and this does not necessarily produce the
maximal bandwidth reduction.
Experimenting with sparser matrices we observed that often the two methods
produce similar results.

We remark that, also in this application, the presence of a multiple Fiedler
value may constitute a problem.
For example, we were not able to correctly process the Matlab test matrix
\ttt{bucky} (the connectivity graph of the Buckminster Fuller geodesic dome),
because the associated Laplacian possesses a triple Fiedler value.

%%%%%%%%%%%%%%%%%%%%%%%%%%%%%%%%%%%%%%%%%%%%%%%%%%%%%%%%%%%%%%%%%%%%%
\section{Conclusions}\label{sec:last}

In this paper we present a new Matlab toolbox principally aimed to the
solution of the seriation problem, but which can be applied to other related
problems.

It is based on a spectral algorithm introduced in~\cite{atkins1998spectral},
and contains an implementation of PQ-trees as well as some tools for their
manipulation, including an interactive visualization tool.
The implemented algorithm includes the possibility to choose between a small
scale and a large scale algorithm for the computation of the Fiedler vector,
and to detect equal components in the same vector according to a chosen
tolerance. Further, a parallel version of the method is provided.

We also point out the importance of the presence of multiple Fiedler values, a
problem which has not been considered before in the literature and which has a
significant influence on the computation of an approximate solution to the
seriation problem.

The use of the toolbox is illustrated by a few practical examples, and its
performance is investigated through a set of numerical experiments, both of
small and large scale.

%%%%%%%%%%%%%%%%%%%%%%%%%%%%%%%%%%%%%%%%%%%%%%%%%%%%%%%%%%%%%%%%%%%%%
\section{Acknowledgements}

We would like to thank Matteo Sommacal for pointing our attention to the
problem of seriation and to its application in archaeology.
The paper~\cite{ps05}, which he coauthored, was our principal source of
information when the research which lead to this paper started.

\bibliographystyle{plain}	% mathematics and physical sciences
\bibliography{bibliogr}

\end{document}